\documentclass[11pt]{article}

\usepackage{fullpage}
\usepackage{graphicx}

\usepackage{epsf}
\usepackage{epsfig}
\usepackage{amsfonts}
\usepackage{amssymb}
\usepackage{latexsym}
\usepackage[all]{xy}
\usepackage[boxed,linesnumbered]{algorithm2e}

\newtheorem{theorem}{Theorem}
\newtheorem{fact}{Fact}
\newtheorem{note}{Note}
\newtheorem{definition}[theorem]{Definition}
\newtheorem{proposition}[theorem]{Proposition}
\newtheorem{lemma}[theorem]{Lemma}
\newtheorem{claim}[theorem]{Claim}
\newenvironment{proof}{\noindent{\bf Proof:}}{\qed\smallskip}
\newtheorem{corollary}[theorem]{Corollary}

\newcommand{\Log}{\mbox{{\sf L}}}

\newcommand{\PosSLP}{\mbox{{\sf PosSLP }}}

\newcommand{\FTCz}{\mbox{\sf FTC}$^0$}
\newcommand{\TCLL}{\mbox{\sf TCLL}}
\newcommand{\FTCLL}{\mbox{\sf TCLL}}
\newcommand{\FOLL}{\mbox{\sf FOLL}}

\newcommand{\CeqNC}{\mbox{{\sf C}$_{=}${\sf NC}$^1$}}

\newcommand{\NP}{\mbox{{\sf NP}}}

\newcommand{\GNC}{\mbox{{\sf GapNC}$^1$}}

\newcommand{\AC}{\mbox{{\sf AC}}}

\newcommand{\ACz}{\mbox{{\sf AC}$^0$}}
\newcommand{\FACz}{\mbox{{\sf FAC}$^0$}}

\newcommand{\TCz}{\mbox{{\sf TC}$^0$}}

\newcommand{\SLP}{\mbox{{\sf SLP}}}

\newcommand{\PP}{\mbox{{\sf PP}}}
\newcommand{\PH}{\mbox{{\sf PH}}}

\newcommand{\qed}{\rule{7pt}{7pt}}

\newcommand{\comment}[1]{}

\SetKwInOut{Input}{Input}
\SetKwInOut{Output}{Output}
\SetKwInOut{Promise}{Promise}

\title{Computing  Bits of Algebraic Numbers}
\author{Samir Datta \\
\texttt{sdatta@cmi.ac.in}\\
Chennai Mathematical Institute\\
  Chennai, India\\
\and Rameshwar Pratap \\
\texttt{rameshwar@cmi.ac.in}\\
Chennai Mathematical Institute\\
  Chennai, India\\
}

\begin{document}
\maketitle
\begin{abstract}
We initiate the complexity theoretic study of the problem of computing the
bits of (real) algebraic numbers. 
This extends the work of Yap on computing the bits of transcendental numbers like
$\pi$, in Logspace.

Our main result is that computing a bit of a fixed real algebraic number is in
\CeqNC$\subseteq \Log$ when the bit position has a verbose (unary) 
representation and in the counting hierarchy when it has a succinct (binary) 
representation.

Our tools are drawn from elementary analysis and numerical analysis, and include
 the Newton-Raphson method.  The proof of our main result
is entirely elementary, preferring to use the elementary Liouville's theorem
over the much deeper Roth's theorem for algebraic numbers.

We leave the possibility of proving non-trivial lower bounds for the
problem of computing the bits of an algebraic number given the bit position
in binary, as our main open question. In this direction we show very
limited progress by proving a lower bound for \emph{rationals}.
\end{abstract}
\newpage
\section{Introduction}
Algebraic numbers are (real or complex) roots of finite degree polynomials with
integer coefficients. Needless to say, they are fundamental objects and play
in important part in all of Mathematics. 


The equivalence between reals with recurring binary expansions (or expansions in
any positive integral radix) and rationals is easy to observe. Thus computing
the bits of fixed rationals is computationally uninteresting. However, the
problem becomes interesting if we focus on irrational real numbers. 
Computability of such numbers heralded the birth of of Computer Science in Turing's landmark
paper \cite{turing} where
the computability of the digits of irrationals like $\pi, e$ is first addressed.

Building on the surprising Bailey–Borwein–Plouffe (BBP) formula 
\cite{BaileyBP} for $\pi$, Yap \cite{Yap}
shows that certain transcendental numbers
such as $\pi$ have binary expansions computable in a small complexity class
like deterministic logarithmic space. Motivated by this result we seek to 
answer the corresponding question for algebraic numbers. The answers we get
turn out to be unsatisfactory but intriguing in many respects. In a nutshell,
we are able to show only a very weak upper bound to the ``succinct'' 
version of the problem and virtually no lower bounds. This gap between best 
known (at least to our knowledge) upper and lower bounds easily beats other 
old hard-to-crack chestnuts such as graph isomorphism and integer 
factorization.

\subsection{Versions of the Problem}
The problem 
as stated in \cite{Yap} asks for the $n$-th bit of the (infinite) binary
sequence of an irrational real given $n$ in \emph{unary}. 
The succinct version of the problem asks for the $n$-th bit, given
$n$ in \emph{binary}. This version of the problem is naturally much
harder than the ``verbose'' version. We can solve the verbose version in \CeqNC,
a subclass of  logspace. For the succinct version of the problem we are unable
to prove a deterministic polynomial or even a non-deterministic polynomial
upper bound. The best we can do is place it at a finite level in the 
counting hierarchy (which includes the computation of the permanent
at its first level). Even more surprising is that we are unable to prove 
any non-trivial lower bound for any irrational algebraic number.
Intriguingly, we can prove Parity and in general $\ACz[p]$ lower bounds for 
computing specific \emph{rationals}.

\subsection{Previous Proof Techniques}
In his article \cite{Yap}, Yap used a BBP like \cite{BBP} series to prove a
logspace upper bound for the (verbose version of) computing the bits of $\pi$. 
At the core of that
argument is the concept of bounded irrationality measure of $\pi$ 
which intuitively measures how inapproximable $\pi$  is, by rationals.
Roughly, the BBP-like series was used to approximate $\pi$ by rationals
and then argue, via the bounded irrationality measure, that, since there aren't
too many good approximations to $\pi$, the computed one must match
the actual expansion to lots of bit positions.

\subsection{Our proof technique}
Further progress was stymied by the extant ignorance of BBP like 
series for
most well-known irrationals. Our crucial observation is that approximating
an irrational can be accomplished by means other than a BBP-like series for
instance by using Newton-Raphson. Bounded irrationality measure for algebraic
numbers follows by a deep theorem of Roth \cite{Roth}. But we show that we
can keep our proof elementary by replacing Roth's theorem by Liouville's
Theorem \cite{Liou} which has a simple and elementary proof (see e.g. 
\cite{Shid}).

An upper bound on the succinct version of the problem follows by observing that
Newton Raphson 
can be viewed as approximating the algebraic number by a rational which is 
the ratio of two \emph{Straight Line Programs} or \emph{SLP}'s.
Allender et al. \cite{AllenderEtAl} show how to compute bits of a single SLP in the
Counting Hierarchy. We extend their proof technique to solve the problem of
computing a bit of the ratio of two SLP's in the Counting Hierarchy.

\subsection{Related Work and Our Results}
Yap \cite{Yap} showed that the bits of $\pi$ are in Logspace. This was the
origin of this endeavour and we are able to refine his result to the following:
\begin{theorem}\label{thm:bbpIrrationals}
Let $\alpha$ be a real number with bounded irrationality measure, which can be expressed
as a convergent series and further the the $m^{th}$ term (for input $m$ in unary) is in \FTCz.
\comment{
Let $\alpha$ be a real number satisfying the following:
\begin{enumerate}
 \item $\alpha$ has bounded irrationality measure.
 \item $\alpha$ can written as convergent series 
 \item the $n^{th}$ term (for input $n$ in unary) is in \FTCz.
\end{enumerate}
}
Then the $n^{th}$ bit of $\alpha$ can be computed by 
a \TCz\, circuit for $n$ in unary and in $\PH^{\PP^{\PP}}$ for 
$n$ in binary.
\end{theorem}
In particular, the above inclusions hold for $\pi$.

Somewhat paradoxically we get slightly weaker bounds for algebraic numbers.
As our main result, we are able to show that (for an explanation of the
complexity classes used in the statement please see the next section):
\begin{theorem}\label{thm:main}
 Let $p$ be a fixed univariate polynomial of degree $d$, having integer 
coefficients. Then, $n^{th}$ bit of each real root of $p$ can be computed in 
\CeqNC$\cap$\TCLL, if $n$ is given in unary, and in $\PH^{\PP^{\PP}}$ if 
$n$ is given in binary.
\end{theorem}
Roughly twenty five years ago,
Ben-Or, Feig, Kozen and Tiwari \cite{BFKT} studied the problem of finding
the roots of a univariate polynomial of degree $n$ and $m$ bit coefficients,
under the promise that \emph{all roots are real}. Under this assumption they
are able to show that approximating the roots to an additive error of
$2^{-\mu}$ for an integer $\mu$ (which is presumably specified in unary)
is in $\mathsf{NC}$. Notice that while their result concerns non-constant
algebraic numbers it does not involve finding the bits of the algebraic numbers 
only approximating them. Also, since they only achieve unary tolerance
their claimed upper bound of $\mathsf{NC}$ is not better than the
\CeqNC $\subseteq \mathsf{NC}$. Also their method does not work if the
all-real-roots promise is not satisfied.

\subsection{Organization of the paper}
In Section~\ref{sec:prelims} we start with pointers to relevant complexity
classes and more importantly known results from elementary analysis that we
will need in our proofs. In Section~\ref{sec:compComplexity} we provide
upper bounds on the complexity of composing bivariate polynomials. This is
used in the subsequent Section~\ref{sec:quadConv} where we view Newton-Raphson
as an iterated composition of bivariate polynomials. In this section we
prove that the method converges ``quickly'' if its initial point is in a
carefully picked interval and that we can efficiently identify such intervals.
In Section~\ref{sec:proofs} we make use of the tools we have put together in the
previous sections to prove Theorem~\ref{thm:bbpIrrationals},~\ref{thm:main}.
We also prove a lower bound on rationals in this section.  Finally in 
Section~\ref{sec:concl} we conclude with some open questions.

\section{Preliminaries}\label{sec:prelims}
\subsection{Complexity Theoretic Preliminaries}
We start off by introducing straight line programs.
An arithmetic circuit is a directed acyclic graph with input nodes labeled 
with the constants $0, 1$ or with indeterminates
$X_{1} , \ldots, X_{k}$ for some $k$. Internal nodes are labeled with 
one of the operations $+,-,*$.
A \emph{straight-line program} is a sequence of instructions corresponding to 
a sequential evaluation of an arithmetic circuit.
We will need to refer to standard complexity classes like $\NP, \PP, \PH$
and we refer the reader to any standard text in complexity such as 
\cite{AroraB09}. We will also use circuit complexity classes like 
$\TCz, \GNC$ and we refer the reader to \cite{Vollmer99} for details.

One non-standard class we use is $\TCLL $. This is inspired by the class
\FOLL\ introduced in \cite{BKLM} which is essentially 
the class of languages accepted by a uniform version of
an \FACz\ circuit iterated $O(\log{\log{n}})$ many times with an \ACz-circuit
on top. We obtain a \TCLL-circuit  by adding a \TCz-circuit on top of the
 iterated block of $\FACz$-circuits. The class of languages accepted by 
such circuits constitutes \TCLL.

\subsection{Mathematical Preliminaries}
In order to upper bound the largest  magnitude of roots of a polynomial we 
note the following (e.g. see Chapter~6 Section~2 of \cite{Yap00}):
\begin{fact}(Cauchy)\label{fact:rouche}
Let $p(x) = \sum_{i=0}^{d}{a_i x^i}$ be a polynomial. Then every root of 
$p(x)$ is smaller in absolute value than:
\[
M_p = 1 + \frac{1}{|a_d|}{\max}_{i=0}^{d-1}{|a_i|}
\]
 We can consider $[-M_p,M_p]$ as possible solution range which contains all the real roots.
\end{fact}

\begin{fact}\label{fact:TaylorSeries}
The Taylor (see \cite{Taylor}) series of a real (complex) function $f(x)$ that is infinitely differentiable in a neighborhood
of a real (complex) number $a$ is the power series
\[
 \sum_{n=0}^{\infty} \frac{f^{(n)}(a)}{n!}(x-a)^n
\]
where $f^{(n)}(a)$ denotes the $n^{th}$ derivative of $f$ evaluated at the point $a$. 
\end{fact}

We will need to lower bound the minimum distance between the roots of a polynomial or the
so called \emph{root separation}. We use the following version of the Davenport-Mahler
theorem (see e.g. Corollary 29 (i) of Lecture VI from Yap~\cite{Yap00} for details of notation
and proof):
\begin{fact}\label{fact:rootSep}(Davenport-Mahler)
The separation between the roots of a univariate polynomial $p(x)$ of degree $d$ is at least:
\[
\sqrt{3|disc(p)|}||p||_2^{-d+1}d^{-(d+2)/2}
\]
\end{fact}

Here, $disc(p)$ is the \emph{discriminant} of $p$ and $||p||_2$ is the $2$-norm of the coefficients
of $p$. Further notice that a lower bound approximation to this bound can be computed in \FTCz\,
for constant polynomials $p$ (since it involves computing lower bound \emph{approximations} for
radicals and powers of constants and cosntant determinants).
\comment{
We will also need to upper bound the minimum distance between the roots of two
polynomials $p(x) = \sum_{i=0}^d{a_i x^i}$, $q(x) = \sum_{i=0}^{e}{b_i x^i}$. 
Let $S_{p,q}$ be the \emph{Sylvester matrix} of the two polynomials of dimension
$(d + e) \times (d+e)$ given by:
\[
S_{p,q}(i,j) = \left\{ \begin{array}{ll}
				a_{d + i - j} & \mbox{ if } 1 \leq i \leq j \leq d + i \leq d + e\\
				b_{i - j} & \mbox{ if } 1 \leq i - e \leq j \leq i \leq d + e \\
				0 & \mbox{ otherwise }
			\end{array}
		\right.
\]
The resultant $R(p,q)$ of $p,q$ is defined to be the determinant of the 
Sylvester matrix \cite{Sylvester}. The following is well known 

\begin{fact}
\[
R(p,q) = \prod_{(x,y) : p(x) = 0 \wedge q(y) = 0}{(x-y)}
\]
\end{fact}

Thus if $r_{p,q} = min_{(x,y): p(x) = 0, q(y) = 0}{|x - y|}$ we have
$r_{p,q}^{de} \leq R(p,q)$ implying that:
\begin{proposition}\label{prop:rootDist}
For fixed polynomials $p,q$:
\[
r_{p,q} \leq (R(p,q))^\frac{1}{de}
\]
Thus, we can compute an upper bound on $r(p,q)$ in \FTCz.
\end{proposition}
\begin{proof}
It suffices to find $2^{\lfloor\frac{\log{R(p,q)}}{de}\rfloor}$
which is easy since the exponent is a constant. 
\end{proof}
}

We will also need an upper bound on the magnitude of the derivative of a
univariate polynomial in open interval $(a,b)$. This is given
by the so called Markoff's Theorem \cite{Markoff} (see also \cite{Ore}):
\begin{fact}\label{fact:markoff}
Let $p(x)$ be a degree $d$ polynomial satisfying,
\[
\forall{x \in (a,b),}\, |p(x)| \leq M_{(a,b)}
\]
Then the derivative $p'(x)$  satisfies:
\[
\forall{x \in (a,b),}\, |p'(x)| \leq M'_{(a,b)} = \frac{d^2 M_{(a,b)}}{b-a}
\]
\end{fact}

\subsubsection{Removing Rational and Repeated Roots}
The rational roots of a polynomial can be dealt with
the aid of the following, easy to see, fact:
\begin{fact}
The rational roots of a monic polynomial (i.e. highest degree coefficient is $1$)
with integer coefficients are integers.
\end{fact}
To make $p(x) = \sum_{i=0}^d{a_ix^i}$ monic, just substitute $y = a_dx$ in
$a_d^{d-1}p(x)$ to obtain a monic polynomial $q(y)$. Iterating over all integers in the
range given by Fact~\ref{fact:rouche} we can find and eliminate the integer roots
of $q(y)$ and therefore the corresponding rational roots of $p(x)$.

In general, a polynomial will have repeated roots. If this is the case 
the repeated root will also be a root of the derivative. Thus it suffices
to find the gcd $g$ of the given polynomial $p$ and its derivative $p'$,
since we can recursively find the roots of $p/g$ and $g$. As we will be
focusing only on fixed polynomials the gcd computation and the division
will be in \FTCz.

We will actually need a more stringent condition on the polynomials - i.e. 
$p$ does not share a root with its derivative $p'$ and even with its
double derivative $p''$. Notice that the above procedure does not guarantee
this. For example if $p(x) =x(x^2-1)$, then $p'(x)=3x^2-1$ and $p''(x)=6x$.
Thus, $p$ and $p''$ share a root but $p$ and $p'$ don't.

So we will follow an iterative procedure in which we find the gcd of $p,p'$
to get two polynomials $p_1 = p/(p,p'), p_2 = (p,p')$ (denoting gcd of $p,q$
by $(p,q)$). For $p_1$ we find
the gcd $(p_1, p''_1)$ and set $p_3 = p_1/(p_1,p''_1), p_4 = (p_1,p''_1)$.
Now we recurse for the $3$ polynomials $p_2, p_3, p_4$ (whose product is $p$).
Notice that the recursion will bottom out when some gcd becomes a constant.
As a result of the above discussion we can assume hereafter that $p$ does
not share a root with either $p'$ or $p''$.
\subsubsection{Good Intervals}
\begin{definition}\label{def:goodInterval}
Fix an interval $I = [a,b]$ of length $|I| = b - a $.
We will call the interval $I$ good for an integral polynomial $p$ if
it contains exactly one root (say $\alpha$) of $p$ and no root of $p'$ and $p''$.
\end{definition}

\begin{lemma}\label{lem:findGoodInterval}
If  $p$ is a polynomial of degree $d$ such that $p$ doesn't share a 
root with its derivative and double derivative then there exist 
$\delta$ such that all intervals of length less than $\delta$ contain at
most one root of $p$ and if they contain a root of $p$ then they
do not contain any roots of $p'$ or $p''$.
 As a consequence we can find good intervals $I$ in \FTCz.
\end{lemma}
\begin{proof}
 Using Fact~\ref{fact:rootSep} compute $\delta$ i.e. the
  minimum distance between roots of $p$ and that of $p'p''$.
(let $p_1,p_2$ be $p',p''$ with the roots common with $p'',p'''$ respectively,
removed; then a lower bound on the root separation of $pp_1p_2$, does the job).
Partition the interval containing
all the roots (which can be inferred from Fact~\ref{fact:rouche}) into
sub-intervals of length $\delta$.  If such an interval contains
a root of $p$ then it contains precisely one root of $p$ and no roots of 
$p'$ or $p''$.  By checking that signs of $p$ are opposite at the end points
 we can identify good intervals.
\end{proof}

\subsection{From approximation to exact computation}
We will crucially use the following theorem (see e.g. Shidlovskii\cite{Shid} or Yap \cite{Yap00})
\begin{fact}\label{fact:liouville} (Liouville's Theorem)
If $x$ is a real algebraic number of degree $d \geq 1$, then there
exists a constant $c = c(x) > 0$ such that the following inequality
holds for any $\alpha \in \mathbb{Z}$ and $\beta \in \mathbb{N}$, $\alpha/\beta \neq x$:
\[
\left| x - \frac{\alpha}{\beta}\right| > \frac{c}{\beta^d}
\]
\end{fact}
The rest of this subsection is an adaptation of the corresponding material
in Chee Yap's paper on computing $\pi$ in \Log. The primary difference being
that we choose to pick the elementary Liouville's Theorem for algebraic
numbers instead of the advanced arguments required for bounding the
irrationality measure of $\pi$. We could throughout replace the use of
Liouville's theorem by the much stronger and deeper Roth's theorem but
prefer not to do so in order to retain the elementary nature of the arguments.

\begin{definition}
Let $x$ be a real number. Let $\{x\} = x - \lfloor x \rfloor$ be the fractional part of $x$.
Further, let $\{x\}_n = \{2^n x\}$ and $x_n$ be the $n$-th bit after the binary point.
\comment{
We denote by:
\begin{itemize}
\item $\{x\} = x - \lfloor x \rfloor$ the fractional part of $x$.
\item $\{x\}_n = \{2^n x\}$ for $n$ non-negative integer.
\item $x_n$ the $n$-th bit after the binary point.
\end{itemize}
}
\end{definition}

It is clear that $x_n = 1$ iff $\{x\}_{n-1} \geq \frac{1}{2}$. For
algebraic numbers we can sharpen this: 
\begin{lemma}\label{lem:yap}(Adapted from Yap\cite{Yap})
Let $x$ be an irrational algebraic number of degree $d$ and let $c = c(x)$
be the constant guaranteed by Liouville's theorem. Let
 $\epsilon_n = c2^{-(d-1)n - 2}$ then for $n$ such that
$\epsilon_n < \frac{1}{4}$ we have:
\begin{itemize}
\item $x_n = 1$ iff $\{x\}_{n-1} \in (\frac{1}{2} + 2\epsilon_n, 1 - 2\epsilon_n)$.
\item $x_n = 0$ iff $\{x\}_{n-1} \in (2\epsilon_n, \frac{1}{2} - 2\epsilon_n)$.
\end{itemize}
\end{lemma}
\begin{proof}
Taking $\beta = 2^{n}$ in Liouville's theorem we get:
$\left| x - 2^{-n}\alpha \right| > \frac{c(x)}{2^{dn}}$
i.e.,
$\left| 2^{n-1} x - \frac{\alpha}{2} \right|  > \frac{c(x)}{2^{d(n-1)+1}} = 2\epsilon_n.$
\end{proof}

Consequently, we can find successive approximations $\{S_m\}_{m \in \mathbb{N}}$ such that the error terms $R_m = x - S_m$ are small enough 
(as described below). $x_n$ is just the $n$-th bit of $S_m$.

\section{Complexity of Composition}\label{sec:compComplexity}
We investigate the complexity of composing polynomials. This will be useful
when we use the Newton-Raphson method to approximate roots of polynomials
since Newton-Raphson can be viewed roughly as an algorithm that 
iteratively composes polynomials.
\begin{definition}
A univariate polynomial with integer coefficients is, an \emph{integral} 
polynomial. Any integral polynomial when evaluated on a rational value 
$\frac{\alpha}{\beta}$, where $\alpha,\beta$ are integers, can be expressed 
as the ratio of two bivariate polynomials in $\alpha,\beta$ called the 
\emph{ratio} polynomials of the integral polynomial.
\end{definition}
\begin{definition} Let $p$ be an integral polynomial. For a positive integer $t$,
the $t$-\emph{composition} of the polynomial denoted by $p^{[t]}$ is 
defined inductively as: $p^{[1]}(x) = p(x)$ and $p^{[t+1]}(x) = p^{[t]}(p(x))$.
\end{definition}
\begin{definition}
Let $f,g$ be a pair of bivariate polynomials. For a positive integer $t$
define the $t$-bicomposition of $(f,g)$ to be the pair of
bivariate polynomials $(F^{[t]},G^{[t]})$ as follows:
\[
F^{[1]}(\alpha,\beta) = f(\alpha,\beta), G^{[1]}(\alpha,\beta) = g(\alpha,\beta),
\]
and,
\[
F^{[t+1]}(\alpha,\beta) = f(F^{[t]}(\alpha,\beta), G^{[t]}(\alpha,\beta)),
\]
\[
G^{[t+1]}(\alpha,\beta) = g(F^{[t]}(\alpha,\beta), G^{[t]}(\alpha,\beta)),
\]
\end{definition}

The following is a direct consequence of the definitions:
\begin{proposition}\label{prop:ratioOfCompositionIsBicomposition}
The ratio polynomials of the $t$-composition of an integral polynomial $p$ are 
exactly the $t$-bicompositions of the ratio polynomials of $p$.
\end{proposition}

\begin{definition}\label{def:Ccomputable}
For an arithmetic complexity class $\mathcal{C}$ containing \FTCz \,\ we call an 
integral polynomial $\mathcal{C}$-computable if its ratio polynomials (viewed as
functions that take in the bit representations of its two (constant)  arguments 
as inputs and output an integer value) are in $\mathcal{C}$.
\end{definition}

Here we consider upper bounds on the complexity of computing compositions of 
fixed integral polynomials.
We first prove that:
\begin{lemma}\label{lem:compositionGapNCo}
Let $p$ be a fixed integral polynomial, then given $n$ in unary the
$l$-bicomposition (for $l = O(\lceil\log{n}\rceil)$) of $p$ is computable in \GNC.
\end{lemma}
\begin{proof}
From Definition~\ref{def:Ccomputable} and 
Proposition~\ref{prop:ratioOfCompositionIsBicomposition}, it suffices to prove
that the bicomposition of the ratio polynomials of $p$ is computable in
\GNC. 
Given a rational number as the bits of its numerator and denominator we can
first obtain the arithmetic values of its numerator and denominator in \GNC. 
Now we are done modulo the following claim.
\end{proof}
\begin{claim}
Let $f$ be a fixed bivariate polynomial with integral coefficients and let
$\alpha,\beta$ be two integers, then there is a constant depth arithmetic
circuit that takes $\alpha,\beta$ as inputs and outputs $f(\alpha,\beta)$.
\end{claim}
\begin{proof}(of Claim) The depth of the circuit is seen to be bounded by 
$O(\lceil\log{d}\rceil)$ where $d$ is the degree of the polynomial - we just
need to find $\alpha^i \beta^j$ by a tree of height 
$max(\lceil\log{i}\rceil, \lceil\log{j}\rceil) + 1$ multiply it with the
coefficient and add up the results by a tree of depth $\lceil\log{(d+1)}\rceil$.
Since degree is a constant we are done. 
\end{proof}

We now prove an orthogonal bound on the complexity of compositions.
But before that we need a small lemma:
\begin{lemma}\label{lem:reachLogDepth}
Suppose $G$ is a layered graph of width $O(n)$ and depth $O(\log{n})$ then
reachability (from a vertex in the first layer to a vertex in the last layer) 
can be done by an \AC-circuit of depth $O(\log{\log{n}})$
\end{lemma}
\begin{proof}
It suffices to show that the reachability between two layers which are
separated by another layer is in \ACz, which is clear.
\end{proof}
\begin{note}
In fact, from the proof it is clear that,
if $G$ contains $O(\log{n})$ identical layers then this reachability
is in the class $\mathsf{FOLL}$ defined in \cite{BKLM}.
\end{note}
\begin{lemma}\label{lem:compositionTCLL}
Let $p$ be a fixed integral polynomial, then given $n$ in unary the
$l$-bicomposition (for $l = O(\lceil\log{n}\rceil)$) of $p$ is computable in 
\TCLL.
\end{lemma}
\begin{proof}
Notice that if $\alpha,\beta$ are some fixed integers, the value of the
$l$-bicomposition of the ratio functions on $\alpha,\beta$ is bounded
by an $n^{O(1)}$ bit integer. It suffices to compute the value of this
composition modulo all $O(\log{n})$ bit primes, since we can do Chinese
Remaindering in \TCz. Fix an $O(\log{n})$ bit prime $q$ and construct the
following bipartite graph $H_q$ on vertices $S,T$ (where $|S| = |T| = q^2$ )
and both $S,T$ consist of pairs $ab$, $a,b \in \{0,\ldots,q-1\}$.
$(ab,a'b')$ is an edge iff $f(a,b) \equiv a'\bmod{q}$ and
$g(a,b) = b'\bmod{q}$ where $f,g$ are the ratio functions
of $p$. Further, let $G_q$ be the layered graph obtained by taking 
$l$ layers of $H_q$. Then it is clear that reachability in $G_q$ from the
first to the last layer is exactly equivalent to the values of the
$l$-compositions modulo $q$. We are done with the aid of 
Lemma~\ref{lem:reachLogDepth} and \cite{HAB}.
\end{proof}

The following is a consequence of the definitions and of \cite{HAB}.

\begin{lemma}\label{lem:unaryApproxRatio}
If $p$ is an integral polynomial which is $\mathcal{C}$-computable
(\FTCz $\subseteq \mathcal{C}$), then on input $m,n$ in unary,
where $m > n$, we can obtain the $n$-th bit of some number that differs from
$p(\frac{\alpha}{\beta})$, by at most $2^{-m}$ in $\mathcal{C}$.
\end{lemma}

Now we describe the binary analog of the above lemma. 
\begin{note}
In the remaining part of this section we denote polynomially bounded
\emph{integers} by lower case letters e.g. $n,t$. We denote those with polynomial 
number of bits  by uppercase letters e.g. $N,T$.
Finally we denote those with exponentially many bits by 
calligraphic letters e.g. $\mathcal{N}, \mathcal{D}$. This notation 
does not apply to rationals like $u,\sigma$.
\end{note}
\begin{lemma}\label{lem:binaryApproxRatio}
Let $\mathcal{N}$ and $\mathcal{D}$  be the outputs of two \SLP's.
Computing the
$N^{th}$ (where $N$ is input in binary) bit of an approximation (accurate 
up to an additive error of $2^{-(N+1)}$) of $\frac{\mathcal{N}}{\mathcal{D}}$  
is in $\PH^{\PP^{\PP}}$.
\end{lemma}

\begin{proof}
 We will compute an under approximation of $\frac{\mathcal{N}}{\mathcal{D}}$ 
with error less than $2^{-(N+1)}$.

Let $u=1-\mathcal{D}{2^{-T}}$ where $T \geq 2$ is an integer such that 
$2^{T-1}\leq \mathcal{D} < 2^T$.  Hence $|u|\leq \frac{1}{2}$.

Notice that the higher order bit of 
$T$ can be found by using \PosSLP\,: we just need to find an integer $t$
such that $2^{2^{t}} \leq \mathcal{D} < 2^{2^{t+1}}$ and both these questions
are \PosSLP\, questions. Having found $T_i$ a lower bound of $T$ correct up to
the higher order $i$ bits of $T$, i.e. $2^{T_i} \leq \mathcal{D} < 2^{T_i + 2^i}$, we check if $2^{T_i + 2^{i-1}} \leq \mathcal{D}$ and update $T_{i-1}$ to
$T_i + 2^{i-1}$ iff the inequality holds (and $T_{i-1} = T_i$ otherwise).
Thus by asking a polynomial number of \PosSLP\, queries, we can determine 
$T$, so each bit of $T$ is in $\PH^{\PP^{\PP}}$.

Now consider the series 
\[
 \mathcal{D}^{-1}=2^{-T}(1-u)^{-1} = 2^{-T}(1 + u + u^2+...)
\]
Set $\mathcal{D}' =2^{-T}(1 + u + u^2+...u^{N+1})$, then\\
\[
 \mathcal{D}^{-1}-\mathcal{D}' \leq \,\ 2^{-T}\sum_{I  >  N+1} 2^{-I} <\,\ 2^{-(N+1)}
\] 
Now we need to compute $N^{th}$ bit of 
\begin{equation}
\frac{\mathcal{N}}{2^T} \sum_{I=0}^{N+1}(1-\frac{\mathcal{D}}{2^T})^I = \frac{1}{2^{(N+2)T}} \sum_{I=0}^{N+1}\mathcal{N}(2^T -\mathcal{D})^I 2^{(N+1 - I)T}
\end{equation}
We need to compute the $M = {N+(N+2)T}^{th}$ bit of: $\mathcal{Y} = \sum_{I=0}^{N+1}{\mathcal{N}(2^T -\mathcal{D})^I 2^{(N+1-I)T}}$.
Since each term in summation is large and there are exponentially many terms in summation, so we will do computation modulo small primes.
Let $\mathcal{Y}_I$ denote the $I^{th}$ term of summation.

Let $\mathcal{M}_n$ be the product of all odd primes less than $2^{n^2}$. 
For such primes $P$ let $H_{P,n}$ denote inverse
of $\frac{\mathcal{M}_n}{P}\bmod{P}$. Any integer $0 \leq Y_I < \mathcal{M}_n$ 
can be represented uniquely as a list ($Y_{I,P}$), where $P$ runs over the 
odd primes bounded by $2^{n^2}$ and $Y_{I,P} = \mathcal{Y}_I\bmod{P}$.

Define the family of approximation functions $app_n(\mathcal{Y})$ to be 
$\sum_{P}{\sum_{I}{Y_{I,P}H_{P,n}\sigma_{P,n}}}$
where $\sigma_{P,n}$ is the 
result of truncating the binary expansion of $\frac{1}{P}$ after $2^{n^4}$ 
bits. 
Notice that for sufficiently large $n$, and $\mathcal{Y} < \mathcal{M}_n$,
$app_n(\mathcal{Y})$ is within $2^{-2^{n^3}} < 2^{-(N+1)}$ of $\mathcal{Y}/\mathcal{M}_n$
as in the proof of Theorem~4.2 of \cite{AllenderEtAl}.
Continuing to emulate that proof further and using the Maciel-Therien (see \cite{MT}) circuit
for iterated addition, we get the same bound as for \PosSLP\, 
in \cite{AllenderEtAl} viz. $\PH^{\PP^{\PP}}$. Notice that we have a double
summation instead of a single one in \cite{AllenderEtAl}, yet it can be 
written out as a large summation and thus does not increase the depth of the
circuit.
\end{proof}


\section{Establishing Quadratic Convergence}\label{sec:quadConv}
We use the famous Newton-Raphson method to approximate Algebraic Numbers.
The treatment is tailored with our particular application in mind. 
There are some features in the proof (for instance a careful use of Markoff's 
result on lower bounding the derivative of a polynomial) which led us to
prove the correctness and rate of convergence of the method from scratch
rather than import it as a black-box.
\begin{definition}(Newton-Raphson)
Given an integral polynomial $p$ and a starting point $x_0$, recursively
define:
\[
	x_{i+1} = x_i - \frac{p(x_i)}{p'(x_i)},
\]
whenever $x_i$ is defined and $p'(x_i)$ is non-zero.
\comment{
whenever both of the following hold:
\begin{enumerate}
\item $x_i$ is defined
\item $p'(x_i)$ is non-zero
\end{enumerate}
}
\end{definition}

Recall good intervals from Definition~\ref{def:goodInterval}.
\begin{definition}\label{def:errorNewton-Rapson}
Given a good interval $I$ for an integral polynomial $p$, let
$\epsilon_i$ denote the error in the $i^{th}$ iteration of Newton-Raphson
i.e. $\epsilon_i = |x_i - \alpha|$ when starting with $x_0 \in I$.
Notice that $\epsilon_i$ is defined only when $x_i$ is.
\end{definition}

\begin{definition}\label{def:quadConv}
We say that Newton-Raphson converges quadratically (with parameter $M$)
for an integral polynomial $p$ whenever  $M$ is 
a non-negative real such that for any interval $I$ which is good for $p$ 
and of length at most $\min(\frac{1}{4M^2},\frac{1}{4})$, it is the case 
that the errors at consecutive iterations (whenever both are defined)
satisfy 
 $\epsilon_{i+1} \leq M \epsilon_i^2$.
\end{definition}

The following Lemma shows that not only are the errors at all iterations defined
under the assumptions of Lemma~\ref{def:quadConv} but also, that, Newton-Raphson
converges ``quickly''.
\begin{lemma} \label{lem:condQuadConv}
 If Newton-Raphson converges quadratically (with parameter $M$)
for an integral polynomial $p$, then 
for every $i \geq 0$, the $i^{th}$ iterand, $x_i$, is at distance 
at most $\min(\frac{1}{4M^2},2^{-2^{i/2}})$ from the unique root of $p$ in any
good interval $I$ of length $|I| \leq \min(\frac{1}{4},\frac{1}{4M^2})$.
In particular, $x_i \in I$ for every $i \geq 0$.
\end{lemma}
\begin{proof}
We proceed by induction on the number of iterations. For the base case, notice
that $x_0$ is at distance at most 
$\epsilon_0 \leq \min(\frac{1}{4},\frac{1}{4M^2})$ from the root.

Now assume that $\epsilon_i < \min(\frac{1}{4M^2}, 2^{-2^{i/2}})$. Then,
\begin{eqnarray*}
\epsilon_{i+1}  & \leq &  M \epsilon_i^2\\
                &  =   & (M \sqrt{\epsilon_i}) \epsilon_i^{1.5} \\
                & \leq & (M\sqrt{\frac{1}{4M^2}})\epsilon_i^{1.5} \\
                & = &  \frac{1}{2} \epsilon_i^{1.5} \\
		& \leq & 2^{-1} 2^{-1.5\times 2^{i/2}} \\
		& < & 2^{-\sqrt{2}\times 2^{i/2}} \\
		& =  & 2^{-2^{(i+1)/2}}
\end{eqnarray*}
Since $\epsilon_{i+1} \leq \frac{1}{2}\epsilon_i^{1.5}$
and $\epsilon_i^{0.5} < \frac{1}{2^{2^{(i-1)/2}}} < 1$ for $i \geq 0$,
therefore, $\epsilon_{i+1} < \epsilon_i < \frac{1}{4M^2}$ where the second
inequality follows from the inductive assumption. This completes the
proof of the inductive step.
\end{proof}

\begin{lemma}\label{lem:subintervalofI}
 For any integral polynomial $p$ and any good interval $I$ thereof, there 
exists a subinterval $I' \subseteq I$  such that
Newton-Raphson converges quadratically in $I'$.
\end{lemma}
\begin{proof}
Let the unique root of the integral polynomial $p$, contained in $I$,
 be $\alpha$. Thus, $p(\alpha) = 0$.
By Taylor's series 
\[
 p(\alpha) = 0 =  p(x_i) + (\alpha - x_i) p'(x_i) + \frac{1}{2}(\alpha -x_i)^2 p''(\xi_i)
\]
where $\xi_i$ is between $x_i$ and $\alpha$. Rearranging, and using the
equation for $x_{i+1}$
\[
\alpha - x_{i+1}  =  (\alpha - x_i) + \frac{p(x_i)}{p'(x_i)} 
                  =  -\frac{1}{2}\frac{p''(\xi_i)}{p'(x_i)}(\alpha - x_i)^2
\]
On the other hand, by Definition~\ref{def:errorNewton-Rapson} the error in the
${i+1}^{th}$ iteration of Newton-Raphson (whenever defined) is:
\[
  \epsilon_{i+1} =  \left|x_{i+1}-\alpha\right|
		 =  \left|\frac{1}{2}\frac{p''(\xi_i)}{p'(x_i)}\right|\epsilon_i^2
\]
Since $p'$ does not have a root in that interval and $p''$ is finite 
(because $p$ is a polynomial), the right hand side is well-defined.

 In the good interval $I$, $p'$ is monotonic and hence the minimum (and maximum)
value of $p'$ is attained at the end-points of $I$.
 Now we can upper bound the absolute value $|p''|$ using upper bound for 
$p'$ and Markoff's result (Fact~\ref{fact:markoff}).
Let this value be denoted by $\rho_1$ and  the minimum value
 of $|p'|$ by $\rho_2 \neq 0$ ($\rho_2 = 0$ would contradict the 
assumption that $I$ does not contain a root of $p'$).
Now, set $M$ to be $\frac{\rho_1}{2\rho_2}$. 

Partition $I$  into sub-intervals of length $\frac{1}{4M^2}$ and
let $I'$ be the unique subinterval containing a root of $p$ :
i.e. the unique sub-interval such that $p$ takes oppositely
signed values at its end points. It is easy to see that Newton-Raphson
converges quadratically (with parameter $M$) in $I'$.
\end{proof}

\section{Putting it all together}\label{sec:proofs}
We now complete the proofs of Theorem~\ref{thm:main} and Theorem~\ref{thm:bbpIrrationals}.

\begin{proof}(of Theorem~\ref{thm:main})
From Lemma~\ref{lem:findGoodInterval} we can compute a good interval $I$.
Then using Lemma~\ref{lem:subintervalofI} we can find a subinterval of $I$
 such that Newton-Raphson will converge quadratically in this interval.

Since Newton-Raphson converges quadratically, in order to obtain an inverse 
exponential error in terms of $n$, (By Lemma~\ref{lem:condQuadConv}) we need
 $O(\lceil\log{n}\rceil)$  iterations. 
Now by Lemma~\ref{lem:compositionGapNCo} and~\ref{lem:compositionTCLL}, 
along with Lemma~\ref{lem:unaryApproxRatio}, we get that
$O(\lceil\log{n}\rceil)$ compositions of Newton-Raphson 
(taking as initial point, the middle  point of the interval $I'$
obtained from Lemma~\ref{lem:subintervalofI}) can 
be computed in \CeqNC$\cap$\FTCLL. 
Finally Lemma~\ref{lem:yap} ensures that we have computed the
 correct bit value.  The argument for the binary case is analogous and uses
Lemma~\ref{lem:binaryApproxRatio} instead of Lemmas~\ref{lem:compositionGapNCo},
~\ref{lem:compositionTCLL}, and \ref{lem:unaryApproxRatio}.
\end{proof}



\begin{proof} (of Theorem~\ref{thm:bbpIrrationals})
 Let $\alpha$ be a constant have series of the form
$ \alpha=\sum_{k=0}^{\infty}t_k=S_n+R_n $
where we have split the series into a finite sum $S_n=\sum_{k=0}^{n}t_k$ and a remainder series $R_n=\sum_{k=n+1}^{\infty}t_k$.
Each term $t_k$ of series can be written as a rational number of the form
$ t_k=\beta^{-kc}\frac{p(k)}{q(k)}$
where $\beta$ is a real number $p(k),q(k)$ are fixed polynomial with integer coefficients and $c\geq 1$ is an integer.

This series consist of summation of iterated multiplication, division and 
addition which can be computed by \TCz \,\ circuit. Since $\alpha$ has 
bounded irrationality measure so its $n^{th}$ bit can be computed using 
Lemma~\ref{lem:yap}.
\end{proof}

Using the BBP-like series for $\pi$ \cite{BBP} and its bounded
irrationality measure, we get:
\begin{corollary}
 Computing $n^{th}$ bit of $\pi$ is in \TCz and $\PH^{\PP^{\PP}}$,
 given $n$ in unary and binary respectively.
\end{corollary}
\comment{
\begin{proof}
Follows from the BBP-like series for $\pi$ \cite{BBP} and its bounded
irrationality measure.
\end{proof}
}

\subsection{Lower Bound} 

Finally we show the $ Mod_p$ (for any odd prime $p$) hardness of the bits
 of a rational. We still don't have
the proof of any kind of hardness of an irrational algebraic number. 

\begin{lemma}
 For given a odd prime $p$ and an integer $X$ (having binary expansion $b_{n-1}\ldots b_0$) then there exist an integer $N$, whose bits
 are constructible by Dlogtime uniform projections, and a fixed rational number $Q$ such 
that  $N^{th}$ digit in binary expansion of $Q$ is $0$ iff  
$\sum_{i}{b_i} \equiv (0 \bmod{p})$.
 \end{lemma}

\begin{proof}
For a given odd prime $p$ we can find a integer $t$, $0< t <p$ such that $2^t \equiv 1(\bmod{p})$. Such a $t$ exists
because the multiplicative group of integers modulo $p$ is finite.
%
%
Consider the number  $N = \sum_{i=0}^{n-1}{b_i(2^t)^i}$.
Then, $ N \equiv \sum_{i=0}^{n-1}{b_i}(\bmod{p})$, because $2^t \equiv 1(\bmod{p})$.
Now, consider the sum: $ Q = \sum_{N > 0}{\frac{N\bmod{p}}{(2^t)^N}}$ $= \frac{2^t(2^{tp}-2^tp+p-1)}{(2^t-1)^2(2^{tp}-1)}$. 
The $ N^{th}$ digit of $ Q $ is  $0$ iff $\sum_{i}{b_i} \equiv 0 \bmod{p}$.
\end{proof}

\section{Conclusion}\label{sec:concl}
We take the first step in the complexity of Algebraic Numbers. Many questions 
remain.  We focus on fixed algebraic numbers - in general we could consider 
       algebraic numbers defined by polynomials of varying degrees/coefficients.
 We have ignored complex algebraic numbers - they could present new challenges.
 Most importantly, our study is, at best, initial because of the enormous
       gap between lower bounds (virtually non-existent) and the upper bounds.
       Narrowing this gap is one of our future objectives.
\comment{
\begin{itemize}
 \item We focus on fixed algebraic numbers - in general we could consider 
       algebraic numbers defined by polynomials of varying degrees/coefficients.
 \item We have ignored complex algebraic numbers - they could present 
       new challenges.
 \item Most importantly, our study is, at best, initial because of the enormous
       gap between lower bounds (virtually non-existent) and the upper bounds.
       Narrowing this gap is one of our future objectives.
\end{itemize}
} 

\section*{Acknowledgements}
We would like to thank Eric Allender, V. Arvind, Narasimha Chary B, 
Raghav Kulkarni, Rohith Varma and Chee K. Yap for illuminating discussions and 
valuable comments on the draft. We also thank anonymous referees of STACS 2012 
for pointing out an error in the previous version and various stylistic improvements.

\bibliographystyle{plain}
\bibliography{reference}

\begin{thebibliography}{10}

\bibitem{Taylor}
M.~Abramowitz and I.~A. Stegun.
\newblock {\em Handbook of Mathematical Functions: with Formulas, Graphs and
  Mathematical Tables}.
\newblock Dover, New York, 1972.

\bibitem{AllenderEtAl}
Eric Allender, Peter B{\"u}rgisser, Johan~Kjeldgaard Pedersen, and Peter~Bro
  Miltersen.
\newblock On the complexity of numerical analysis.
\newblock {\em SIAM J. Comput.}, 38(5):1987--2006, 2009.

\bibitem{AroraB09}
Sanjeev Arora and Boaz Barak.
\newblock {\em Computational Complexity - A Modern Approach}.
\newblock Cambridge University Press, 2009.

\bibitem{BBP}
David~H. Bailey.
\newblock A compendium of {BBP}-type formulas for mathematical constants.
\newblock Report, Lawrence Berkeley National Laboratory, Berkeley, CA, USA,
  February 2011.

\bibitem{BaileyBP}
David~H. Bailey, Jonathan~M. Borwein, Peter~B. Borwein, and Simon Plouffe.
\newblock The quest for pi.
\newblock {\em The Mathematical Intelligencer}, 19(1):50--57, January 1997.

\bibitem{BKLM}
David A.~Mix Barrington, Peter Kadau, Klaus-J{\"o}rn Lange, and Pierre
  McKenzie.
\newblock On the complexity of some problems on groups input as multiplication
  tables.
\newblock {\em J. Comput. Syst. Sci.}, 63(2):186--200, 2001.

\bibitem{BFKT}
Michael Ben-Or, Ephraim Feig, Dexter Kozen, and Prasoon Tiwari.
\newblock A fast parallel algorithm for determining all roots of a polynomial
  with real roots.
\newblock {\em SIAM J. Comput.}, 17(6):1081--1092, 1988.

\bibitem{Liou}
G.~H. Hardy and E.~M. Wright.
\newblock {\em An Introduction to the Theory of Numbers}.
\newblock Oxford Univ. Press, New York, 5th ed edition, 1979.

\bibitem{HAB}
William Hesse, Eric Allender, and David A.~Mix Barrington.
\newblock Uniform constant-depth threshold circuits for division and iterated
  multiplication.
\newblock {\em J. Comput. Syst. Sci.}, 65(4):695--716, 2002.

\bibitem{MT}
Alexis Maciel and Denis Th{\'e}rien.
\newblock Threshold circuits of small majority-depth.
\newblock {\em Inf. Comput.}, 146(1):55--83, 1998.

\bibitem{Markoff}
A.~Markoff.
\newblock Sur une question pos\'ee par {M}endeleieff.
\newblock {\em Bulletin of the Academy of Sciences of St. Petersburg},
  62:1--24, 1889.

\bibitem{Ore}
Oystein Ore.
\newblock On functions with bounded derivatives.
\newblock {\em Transactions of the American Mathematical Society}, 43(2):pp.
  321--326, 1938.

\bibitem{Roth}
Klaus~Friedrich Roth.
\newblock Rational approximations to algebraic numbers.
\newblock {\em Mathematika. A Journal of Pure and Applied Mathematics},
  2:1--20, 1955.

\bibitem{Shid}
A.~B Shidlovskii.
\newblock {\em Transcendental Numbers}.
\newblock de Gruyter, New York, 1989.

\bibitem{turing}
A.~M. Turing.
\newblock {On Computable Numbers, with an application to the
  Entscheidungsproblem}.
\newblock {\em Proc. London Math. Soc.}, 2(42):230--265, 1936.

\bibitem{Vollmer99}
Heribert Vollmer.
\newblock {\em Introduction to circuit complexity - a uniform approach}.
\newblock Texts in theoretical computer science. Springer, 1999.

\bibitem{Yap00}
Chee Yap.
\newblock {\em Fundamental Problems in Algorithmic Algebra}.
\newblock Oxford University Press, 2000.

\bibitem{Yap}
Chee Yap.
\newblock Pi is in log space.
\newblock manuscript, June 2010.

\end{thebibliography}
\end{document}